\documentclass[aps,prl,twocolumn,groupedaddress,amsmath,superscriptaddress,english]{revtex4}
\usepackage{amsfonts}
\usepackage{amssymb}
\usepackage{amsmath}
\usepackage[dvips]{graphicx}
\usepackage{xcolor}
\usepackage{dcolumn}
\usepackage{bm}
\usepackage{soul}
\usepackage{dsfont}

\definecolor{colorone}{rgb}{1,0.36,0.03}
\definecolor{colortwo}{rgb}{0.54,0.71,0.03}
\definecolor{colorthree}{rgb}{0.01,0.51,0.93}
\definecolor{colorfour}{rgb}{0.47,0.26,0.58}

\def\b{\beta}

\makeatletter
\newcommand{\rmnum}[1]{\romannumeral #1}
\newcommand{\Rmnum}[1]{\expandafter\@slowromancap\romannumeral #1@}
\makeatother

\usepackage{hyperref}
\hypersetup{colorlinks=true,citecolor=blue,linkcolor=blue,filecolor=blue,urlcolor=blue,breaklinks=true,linktocpage=true}

\usepackage{enumitem}
\usepackage{amsmath}
\usepackage{graphicx}
\usepackage{amsfonts}
\usepackage{amssymb}
\usepackage{hyperref}
\usepackage{color}
\usepackage{mathrsfs}
\usepackage{isomath}
\usepackage{amsthm}
\usepackage{epstopdf}
\usepackage{txfonts}
\usepackage{float}

\newcounter{remark}

\usepackage{tikz}
\usetikzlibrary{patterns}
\makeatletter
\def\grd@save@target#1{%
  \def\grd@target{#1}}
\def\grd@save@start#1{%
  \def\grd@start{#1}}
\tikzset{
  grid with coordinates/.style={
    to path={%
      \pgfextra{%
        \edef\grd@@target{(\tikztotarget)}%
        \tikz@scan@one@point\grd@save@target\grd@@target\relax
        \edef\grd@@start{(\tikztostart)}%
        \tikz@scan@one@point\grd@save@start\grd@@start\relax
        \draw[minor help lines,magenta] (\tikztostart) grid (\tikztotarget);
        \draw[major help lines] (\tikztostart) grid (\tikztotarget);
        \grd@start
        \pgfmathsetmacro{\grd@xa}{\the\pgf@x/1cm}
        \pgfmathsetmacro{\grd@ya}{\the\pgf@y/1cm}
        \grd@target
        \pgfmathsetmacro{\grd@xb}{\the\pgf@x/1cm}
        \pgfmathsetmacro{\grd@yb}{\the\pgf@y/1cm}
        \pgfmathsetmacro{\grd@xc}{\grd@xa + \pgfkeysvalueof{/tikz/grid with coordinates/major step}}
        \pgfmathsetmacro{\grd@yc}{\grd@ya + \pgfkeysvalueof{/tikz/grid with coordinates/major step}}
        \foreach \x in {\grd@xa,\grd@xc,...,\grd@xb}
        \node[anchor=north] at (\x,\grd@ya) {\pgfmathprintnumber{\x}};
        \foreach \y in {\grd@ya,\grd@yc,...,\grd@yb}
        \node[anchor=east] at (\grd@xa,\y) {\pgfmathprintnumber{\y}};
      }
    }
  },
  minor help lines/.style={
    help lines,
    step=\pgfkeysvalueof{/tikz/grid with coordinates/minor step}
  },
  major help lines/.style={
    help lines,
    line width=\pgfkeysvalueof{/tikz/grid with coordinates/major line width},
    step=\pgfkeysvalueof{/tikz/grid with coordinates/major step}
  },
  grid with coordinates/.cd,
  minor step/.initial=.2,
  major step/.initial=1,
  major line width/.initial=2pt,
}
\makeatother

\newcommand{\Tr}{\mbox{Tr}}

\def\be{\begin{equation}}
\def\ee{\end{equation}}
\def\ba{\begin{array}}
\def\ea{\end{array}}
\def\Tr{\mathrm{Tr}}

\newcommand{\p}{\mathbf{p}}
\newcommand{\q}{\mathbf{q}}

\newcommand{\x}{\mathbf{x}}
\newcommand{\y}{\mathbf{y}}
\newcommand{\z}{\mathbf{z}}

\usepackage{tikz}
\usetikzlibrary{patterns}
\makeatletter
\def\grd@save@target#1{%
  \def\grd@target{#1}}
\def\grd@save@start#1{%
  \def\grd@start{#1}}
\tikzset{
  grid with coordinates/.style={
    to path={%
      \pgfextra{%
        \edef\grd@@target{(\tikztotarget)}%
        \tikz@scan@one@point\grd@save@target\grd@@target\relax
        \edef\grd@@start{(\tikztostart)}%
        \tikz@scan@one@point\grd@save@start\grd@@start\relax
        \draw[minor help lines,magenta] (\tikztostart) grid (\tikztotarget);
        \draw[major help lines] (\tikztostart) grid (\tikztotarget);
        \grd@start
        \pgfmathsetmacro{\grd@xa}{\the\pgf@x/1cm}
        \pgfmathsetmacro{\grd@ya}{\the\pgf@y/1cm}
        \grd@target
        \pgfmathsetmacro{\grd@xb}{\the\pgf@x/1cm}
        \pgfmathsetmacro{\grd@yb}{\the\pgf@y/1cm}
        \pgfmathsetmacro{\grd@xc}{\grd@xa + \pgfkeysvalueof{/tikz/grid with coordinates/major step}}
        \pgfmathsetmacro{\grd@yc}{\grd@ya + \pgfkeysvalueof{/tikz/grid with coordinates/major step}}
        \foreach \x in {\grd@xa,\grd@xc,...,\grd@xb}
        \node[anchor=north] at (\x,\grd@ya) {\pgfmathprintnumber{\x}};
        \foreach \y in {\grd@ya,\grd@yc,...,\grd@yb}
        \node[anchor=east] at (\grd@xa,\y) {\pgfmathprintnumber{\y}};
      }
    }
  },
  minor help lines/.style={
    help lines,
    step=\pgfkeysvalueof{/tikz/grid with coordinates/minor step}
  },
  major help lines/.style={
    help lines,
    line width=\pgfkeysvalueof{/tikz/grid with coordinates/major line width},
    step=\pgfkeysvalueof{/tikz/grid with coordinates/major step}
  },
  grid with coordinates/.cd,
  minor step/.initial=.2,
  major step/.initial=1,
  major line width/.initial=2pt,
}
\makeatother

\makeatletter
\newcommand*\bigcdot{\mathpalette\bigcdot@{.5}}
\newcommand*\bigcdot@[2]{\mathbin{\vcenter{\hbox{\scalebox{#2}{$\m@th#1\bullet$}}}}}
\makeatother

\usepackage{booktabs}
\usepackage{adjustbox}

\newtheorem{lem}{Lemma}

\newtheorem{definition}{Definition}

\makeatletter
\newcommand*{\rom}[1]{\expandafter\@slowromancap\romannumeral #1@}
\makeatother

\begin{document}

\title{Strong Majorization Uncertainty Relations: Theory and Experiment}

\author{Yuan Yuan}
\affiliation{CAS Key Laboratory of Quantum Information, University of Science and Technology of China, Hefei, 230026, China}
\affiliation{Department of Physics, East China University of Science and Technology, Shanghai, 200237, China}
\affiliation{Synergetic Innovation Center of Quantum Information and Quantum Physics, University of Science and Technology of China, Hefei, Anhui 230026, China}

\author{Yunlong Xiao}
\email{yunlong.xiao@ucalgary.ca}
\affiliation{Department of Mathematics and Statistics, University of Calgary, Calgary, Alberta T2N 1N4, Canada}
\affiliation{Institute for Quantum Science and Technology, University of Calgary, Calgary, Alberta, T2N 1N4, Canada}

\author{Zhibo Hou}
\affiliation{CAS Key Laboratory of Quantum Information, University of Science and Technology of China, Hefei, 230026, China}
\affiliation{Synergetic Innovation Center of Quantum Information and Quantum Physics, University of Science and Technology of China, Hefei, Anhui 230026, China}

\author{Shao-Ming Fei}
\affiliation{School of Mathematical Sciences, Capital Normal University, Beijing 100048, China}
\affiliation{Max Planck Institute for Mathematics in the Sciences, 04103 Leipzig, Germany}

\author{Gilad~Gour}
\affiliation{Department of Mathematics and Statistics, University of Calgary, Calgary, Alberta T2N 1N4, Canada}
\affiliation{Institute for Quantum Science and Technology, University of Calgary, Calgary, Alberta, T2N 1N4, Canada}

\author{Guo-Yong Xiang}
\email{gyxiang@ustc.edu.cn}
\affiliation{CAS Key Laboratory of Quantum Information, University of Science and Technology of China, Hefei, 230026, China}
\affiliation{Synergetic Innovation Center of Quantum Information and Quantum Physics, University of Science and Technology of China, Hefei, Anhui 230026, China}

\author{Chuan-Feng Li}
\affiliation{CAS Key Laboratory of Quantum Information, University of Science and Technology of China, Hefei, 230026, China}
\affiliation{Synergetic Innovation Center of Quantum Information and Quantum Physics, University of Science and Technology of China, Hefei, Anhui 230026, China}

\author{Guang-Can Guo}
\affiliation{CAS Key Laboratory of Quantum Information, University of Science and Technology of China, Hefei, 230026, China}
\affiliation{Synergetic Innovation Center of Quantum Information and Quantum Physics, University of Science and Technology of China, Hefei, Anhui 230026, China}

\begin{abstract}
In spite of enormous theoretical and experimental progresses in quantum uncertainty relations, the experimental investigation of most current, and universal formalism of uncertainty relations, namely majorization uncertainty relations (MURs), has not been implemented yet. A significant problem is that previous studies on the classification of MURs only focus on their mathematical expressions, while the physical difference between various forms remains unknown. First, we use a guessing game formalism to study the MURs, which helps us disclosing their physical nature, and distinguishing the essential differences of physical features between diverse forms of MURs. Second, we tighter the bounds of MURs in terms of flatness processes, or equivalently, in terms of majorization lattice. Third, to benchmark our theoretical results, we experimentally verify MURs in the photonic systems.
\end{abstract}

\maketitle

\textbf{\textit{Introduction.---}}In the quantum world, measurements allow us to gain information from a system, and the action of measurements on quantum systems is fully embraced in the areas of quantum optics, quantum information theories, and quantum communication tasks. It is therefore of great practical interest to study the limitations and precisions of quantum measurements. In taking the measurements on board, however, it appears that quantum mechanics imposes strict limitation on our ability to specify the precise outcomes from incompatible measurements simultaneously, which is known as ``{\it Heisenberg Uncertainty Principle}'' \cite{Heisenberg1927}.

In the context of the uncertainty principle, both variance \cite{Kennard1927, Weyl1927, Robertson1929,Schrodinger1930,Huang2012,Maccone2014,Xiao2016W,Xiao2016M,Xiao2016S,Guise2018,Qiao2018T,Xiao2017I} and entropies \cite{Deutsch1983,Partovi1983,Kraus1987,Maassen1988,Ivanovic1992,Sanchez1993,Ballester2007,Wu2009,Berta2010,Li2011,Prevedel2011,Huang2011,Tomamichel2011,Coles2012,Coles2014,Kaniewski2014,Furrer2014,Li2015,Berta2016,Xiao2016St, Xiao2016QM, Xiao2016U,review,Xiao2018H,Xiao2018Q,Chen2018,Coles2019,Li2019,Wang2019E,Xiao2019CIP,XiaoPhD} are by no reason the most adequate to use. The attempt to find all suitable uncertainty measures has triggered the interest of the scientific community in the quest for a better understanding and exploitation of the precisions of quantum measurements. As previously shown in \cite{PRL,Narasimhachar2016}, any eligible candidate of uncertainty measures should be: (\rmnum{1}) non-negative; (\rmnum{2}) a function only of the probability vector associated with the measurement outcomes; (\rmnum{3}) invariant under permutations; (\rmnum{4}) nondecreasing under a random relabeling. According to these restrict conditions, a qualified uncertainty measure should be a non-negative Schur-concave function, and the {\it majorization uncertainty relations} (MURs) arise from the fact that all Schur-concave functions can, in general, preserve the partial order induced by majorization~\cite{Hardy1929,Partovi2011,Majorization}. Based on the mathematical expressions, the notions of MURs are classified into two categories; that are direct-product MUR (DPMUR) \cite{PRL,JPA} and direct-sum MUR (DSMUR) \cite{PRA,M7}. 
In the original work of \cite{PRA}, the essential differences of mathematical features between DPMUR and DSMUR (i.e. tensor and direct-sum) are compared and analyzed. However, it is fair to say, that our understanding of the physical essences of MURs is still very limited.

In this work, our first contribution, which also reflects the original intention of this work, is to characterize the essential differences of physical features between DPMUR and DSMUR theoretically. More precisely, we show that the difference between these MURs are more than its mathematical expressions, what really matters is the joint uncertainty they represent. DPMUR is identified as a type of {\it spatially-separated joint uncertainty}, and meanwhile DSMUR is recognized as a type of {\it temporally-separated joint uncertainty}. Despite previous developments on MURs, there is still a gap between their optimal bounds and the ones constructed in \cite{PRL,JPA,PRA,M7}. Our second contribution is to fill this gap by applying a technique, called flatness process \cite{Cicalese2002}, which is also known as concave envelope in Mathematics.

Besides theoretical advancements, the experimentally implementations of quantum uncertainty relations are also already of great interest, as they are a pioneering demonstration of the limitations on quantum measurements, and may also inspire breakthrough in modern quantum technologies. So far, the uncertainty relations based on variance and entropies have been successfully realized in various physical systems, including neutronic systems \cite{neutron1,neutron2,neutron3}, photonic systems \cite{photon1,photon2,photon3,photon4,photon5,photon6,photon7}, nitrogen-vacancy (NV) centres \cite{NV}, nuclear magnetic resonance (NMR) \cite{NMR}, and so forth. However, an experimental demonstration of the uncertainty relations given by majorization has never been shown. To boost the experimentally study of the uncertainty relations, it is highly desirable to know how to investigate MURs in a physical system. The third contribution of this work is that we implement the MURs by measuring a qudit state encoded with the path and polarization degree of the freedom of a photon system for the first time.

\begin{figure}[h]
\centering
\begin{tikzpicture}


  \node[] at (-2.7,-0.4) {(a) DPMUR.};

  \draw[very thick,dashed,fill=magenta,opacity=0.2] (-5,3) rectangle (-0.75,0);
  \draw[very thick,fill=black!70,black!70] (-4.8,2.5) rectangle (-3.8,1.8) node[pos=0.5,white] {\large $\Gamma_\rho$};
  \draw[very thick] (-3,2.5) rectangle (-2,1.8);

  \draw[very thick,->] (-2.5,1.95) -- (-2.35,2.4);
  \draw[very thick] (-2.15,1.95) arc (10:170:0.35);

  \draw[very thick,fill=black!70,black!70] (-4.8,1.2) rectangle (-3.8,0.5) node[pos=0.5,white] {\large $\Gamma_\rho$};
  \draw[very thick,->] (-2.5,0.65) -- (-2.35,1.1);

  \draw[very thick] (-2.15,0.65) arc (10:170:0.35);
  \draw[very thick] (-3,1.2) rectangle (-2,0.5);

  \node[] at (-2.5,2.7) {$M$};
  \node[] at (-2.5,0.3) {$N$};

  \draw[very thick,->] (-3.8,2.15) -- (-3,2.15);
  \draw[very thick,->] (-3.8,0.8) -- (-3,0.8);

  \draw[very thick,->] (-2,2.15) -- (-1.15,2.15) -- (-1.15,1.75);
  \draw[very thick,->] (-2,0.8) -- (-1.15,0.8) -- (-1.15,1.2);
  \node[] at (-1.6,2.4) {$a$};
  \node[] at (-1.6,0.55) {$b$};
  \node[] at (-1.15,1.475) {$(a,b)$};

  \node[] at (-3.4,2.4) {$\mathcal H$};
  \node[] at (-3.4,0.55) {$\mathcal H$};

  \begin{scope}[shift={(-0.25,0)}]

  \draw[very thick,dashed,fill=cyan,opacity=0.2] (-0.2,3) rectangle (3.7,0);
  \draw[very thick,fill=black!70,black!70] (0,2.5) rectangle (1,1.8) node[pos=0.5,white] {\large $\Gamma_\rho$};
  \draw[very thick] (1.8,2.5) rectangle (2.8,1.8);

  \draw[very thick,->] (2.3,1.95) -- (2.45,2.4);
  \draw[very thick] (2.65,1.95) arc (10:170:0.35);

  \draw[very thick] (1.8,1.2) rectangle (2.8,0.5);


  \draw[very thick,->] (1,2.15) -- (1.8,2.15) node[pos=0.5,shift={(0,0.22)}] {$\mathcal H$};




  \node[] at (2.3,2.7) {$M/N$};
  \node[] at (2.3,0.85) {$R$};

  \draw[very thick,->] (2.3,1.2) -- (2.3,1.8);
  \draw[very thick,->] (2.8,2.15) -- (3.3,2.15);
  \node[] at (3.48,2.15) {$\p$};
  \node[] at (2.1,1.4) {$0$};
  \node[] at (2.5,1.4) {$1$};

  \node[] at (2,-0.4) {(b) DSMUR.};
  \end{scope}
\end{tikzpicture}
\caption{(color online)  Schematic illustration of the DPMUR (a) and DSMUR (b).}
\label{comp}
\end{figure}
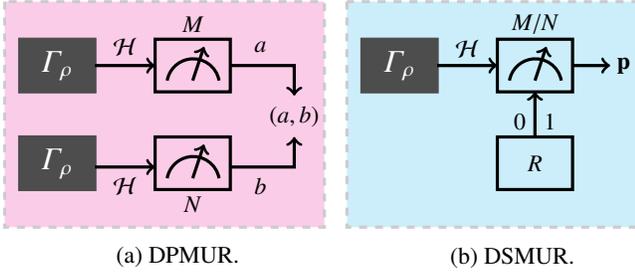

\textbf{\textit{Direct-Product.---}}The construction of DPMUR proposed in~\cite{PRL,JPA} is best formulated as a game, shown in Fig. (\ref{comp}a),  between an experimentalist (Alice) and a referee (Bob) trying to guess the measurement outcomes. More explicitly, the game considered here is as follows: two black boxes $\mathrm{\Gamma}_{\rho}$ are located in different positions, each of them provides a quantum state $\rho$ to Alice and She implements her measurements $M$ and $N$ to the input state separately in each round. Alice knows the measurement outcome from experiments, but she does not know the actual state given to her. By repeating the same procedure a sufficient number of times, Alice derive distinct pairs of measurement outcomes, and the goal of Bob is to guess $k$ distinct pairs of them correctly.

Mathematically, $\mathrm{\Gamma}_{\rho}$ is a preparation channel, generating quantum state $\rho$ on a Hilbert space $\mathcal{H} \cong \mathbb{C}^{d}$~\cite{Kraus1983}. The outcome $a$ of the positive operator valued measure (POVM) $M = \{M_{a}\}$ occurs with probability $p_{a} := \Tr(M_{a}~\rho)$ ($a=1, \ldots, n$). Similarly, we implement the measurement $N$, and denote the corresponding probability distribution by $q_{b} := \Tr(N_{b}~\rho)$ ($b=1, \ldots, m$). We collect the numbers $p_{a}$ and $q_{b}$ into two probability vectors $\p$ and $\q$, respectively.

In the present scheme, the joint uncertainty between $\p$ and $\q$ is captured by the maximal probability of Bob in winning the game. For example, when Alice receives outcome $(a, b)$ from measurements, Bob will have a maximal probability $\max_{\rho}p_{a}q_{b}$ to win. In general, if Alice receive $k$ distinct pairs of outcomes, then the quantum mechanics gives Bob $R_{k}$ chance to win, with
\begin{align}
R_{k} := \max\limits_{I_{k}}\max\limits_{\rho}\sum\limits_{(a,b)\in I_{k}}
p_{a}q_{b},\notag
\end{align}
where $I_{k} \subset [n] \times [m]$ is a subset of $k$ distinct pair of indices. Here $[n] = \{ 1, \ldots, n\}$ is the set of natural numbers ranging from $1$ to $n$, and $k \in [mn]$. Clearly, such guessing game can be reformulated as the following $mn$ inequalities
\begin{align}
\sum\limits_{(a,b)\in I_{k}}
p_{a}q_{b} \leqslant R_{k}. \quad\forall k\in[mn] \notag
\end{align}
A concise approach of expressing the inequalities mentioned above is to use the majorization ``$\prec$'' \cite{Majorization}; A probability vector $\x \in \mathbb{R}^n$ is majorizied by $\y \in \mathbb{R}^n$, i.e. $\x \prec \y$, if and only if $\sum_{j=1}^{k} x^{\downarrow}_{j} \leqslant \sum_{j=1}^{k} y^{\downarrow}_{j}$ for all $1\leqslant k \leqslant n-1$. Here the down-arrow indicates that the components of the vectors are arranged in a non-increasing order. Now we can abbreviate the guessing game into one inequality
\begin{align}\label{dp}
\p\otimes\q \prec \bm{r},
\end{align}
with $\bm{r} :=(R_{1}, R_{2}-R_{1}, \ldots, R_{mn}-R_{mn-1})$. Consequently, the essence of DPMUR is captured by our framework of guessing game, which demonstrates a {\it spatially-separated joint uncertainty}. Note that $R_{k}$ can be in general difficult to calculate explicitly, as they involve an optimization problem. However, the authors of~\cite{PRL} provide us a calculate-friendly bound $\bm{t}$, satisfying $\p\otimes\q \prec \bm{r} \prec \bm{t}$.

Physically, MURs are very general; they encompass the most well-known entropic functions used in quantum information theory, but they are not restricted to these functions. Mathematically, majorization lattice forms a {\it complete lattice}; the optimal bounds for MURs exist. To obtain the optimal bounds, it suffices to perform a standard process (flatness process) $\mathcal{F}$. Hence, the implementation of the process $\mathcal{F}$ on $\p\otimes\q \prec \bm{r} \prec \bm{t}$ could lead to a new relation
\begin{align}\label{fdp}
\p\otimes\q \prec \mathcal{F} (\bm{r}) \prec \bm{r} \prec \mathcal{F} (\bm{t})
\prec \bm{t},
\end{align}
where $\bm{r}$ and $\bm{t}$ are the bounds given in \cite{PRL}. Because of the mathematical properties of flatness process (concave envelope), the vector $\mathcal{F} (r)$ is optimal. However, a major drawback of $\mathcal{F} (\bm{r})$ is that the calculation of $\mathcal{F} (\bm{r})$ is even harder than $\bm{r}$. With the help of flatness process, we also obtain an effectively computable bound $\mathcal{F} (\bm{t})$, which is tighter than the original $\bm{t}$.
We defer the construction of $\bm{t}$, and the rigorous definition of flatness process to the Supplementary Material \cite{SM}.

\begin{figure*}[tbph]
\includegraphics [width=16cm,height=7.6cm]{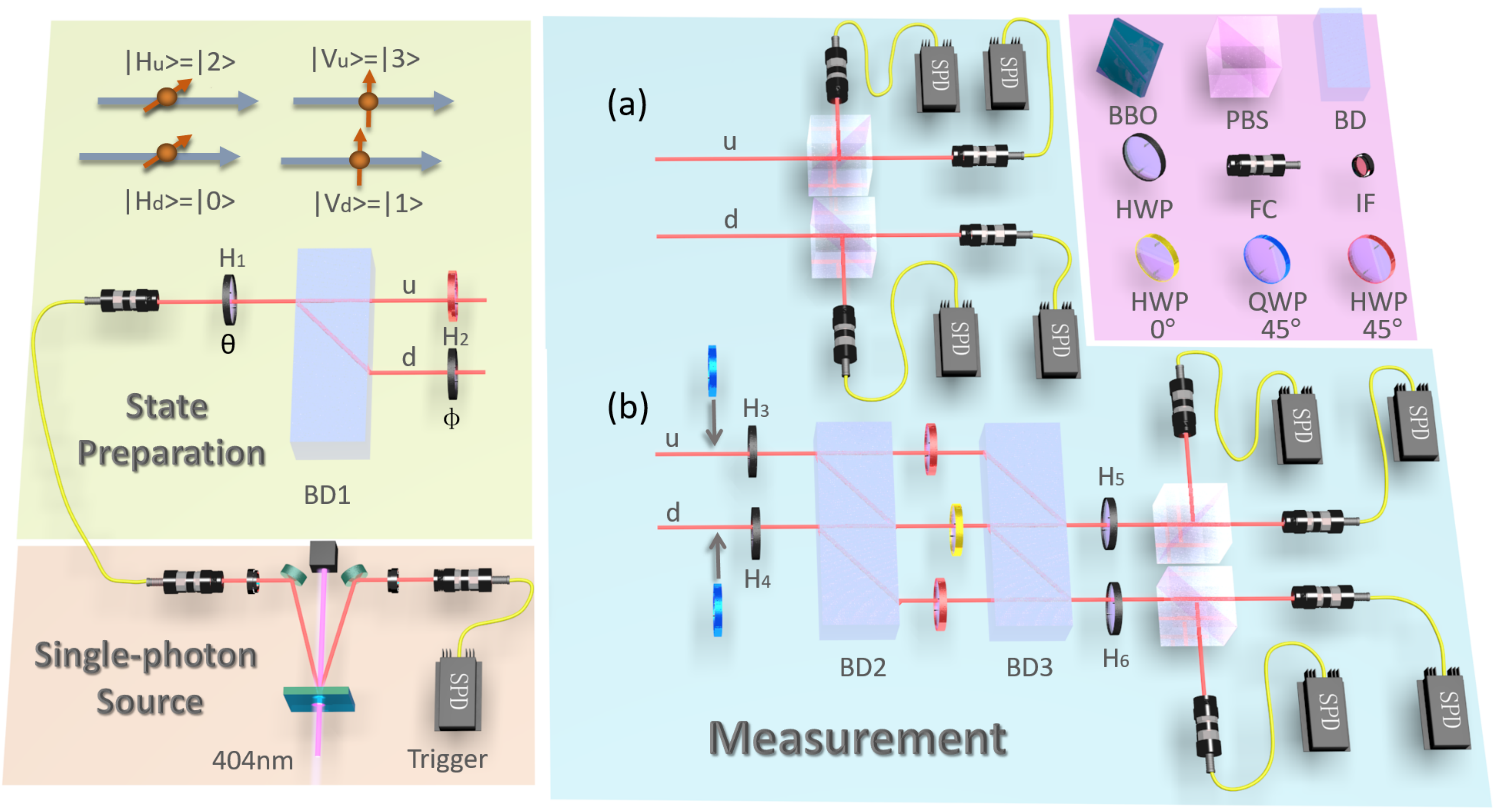}
\caption{(color online) Experimental setup. In the single-photon source module, the photon pairs generated in spontaneous parametric down-conversion are coupled into single-mode fibers separately. One photon is detected by a single-photon detector (SPD) acting as a trigger. In the state preparation module, a qudit is encoded by four modes of the single photon. H and V denote the horizontal polarization and vertical polarization of the photon, respectively. The subscripts u and d represent the upper and lower spatial modes of the photon, respectively. The half-wave plates (H$_{1}$, H$_{2}$) and beam displacer (BD1) are used to generate desired qudit state. In the measurement module, the red HWPs with an angle of $45^{\circ}$ and beam displacers (BDs) comprise the interferometric network to perform the desired measurement; the yellow HWP with an angle of $0^{\circ}$ are inserted into the middle path to compensate the optical path difference between the upper and lower spatial modes. To realize measurement $B$ shown in Eq. (\ref{2measurement}), two quarter-wave plates are need to be inserted in device (b). Output photons are detected by SPDs.}
\label{fig:setup}
\end{figure*}

\textbf{\textit{Direct-Sum.---}}Our protocol of DSMUR combines guessing game with a binary random number generator $R$, shown in Fig. (\ref{comp}b); in each round, the measurement is determined by $R$. More specifically, $R$ outputs number $0$ with probability $\lambda$, and $1$ with probability $1-\lambda$. After receiving $0$, Alice performs $M$, otherwise she implements $N$. Again the goal of Bob is to guess the measurement outcome of Alice. The maximal probability for Bob to guess $k$ outcomes correctly is given by
\begin{align}\label{sk}
S_{k} := \max\limits_{|I|+|J|=k}\max\limits_{\rho}\sum\limits_{\substack{a\in I \subset [n] \\b\in J \subset [m] }}
\left(\lambda p_{a} + (1-\lambda) q_{b}\right) \notag
\end{align}
where $|\bigcdot|$ denotes the cardinality of $\bigcdot$. There exists an efficient way of computing the success probability $S_{k}$. Let us define an operator $G_{c}$ as
\begin{equation}
G_{c}(\lambda) := \left\{
\begin{aligned}
&\lambda M_{c}  ~ & 1 \leqslant c \leqslant n, \\
&(1-\lambda) N_{c-n} ~ & n+1 \leqslant c \leqslant n+m. \notag
\end{aligned}
\right.
\end{equation}
Then the quantity $S_{k}$ becomes
\begin{align}
S_{k}(\lambda) = \max\limits_{|I|=k}
\lambda_{1}\left( \sum\limits_{c\in I \subset [n+m]}G_{c}(\lambda) \right),\notag
\end{align}
where $\lambda_{1}(\bigcdot)$ denotes the maximum eigenvalue of the argument. Now we can conclude our guessing game within one inequality by using majorization; that is
\begin{align}\label{ds}
\lambda\p\oplus(1-\lambda)\q \prec \bm{s}(\lambda),
\end{align}
with $\bm{s}(\lambda):=(S_{1}(\lambda), S_{2}(\lambda)-S_{1}(\lambda), \ldots, S_{m+n}(\lambda)-S_{m+n-1}(\lambda))$. In the framework of DSMUR, classical uncertainty of the random number generator is injected into the guessing game, and as a consequence $\lambda\p\oplus(1-\lambda)\q$ is a hybrid type of uncertainty, mingling both classical and quantum uncertainties. Quite remarkably, the measurements, monitored by $R$, can be implemented in the same position but cannot performed simultaneously, and hence $\lambda\p\oplus(1-\lambda)\q$ reveals a {\it temporally-separated joint uncertainty}. It should be stressed here that the original DSMUR \cite{PRA,M7} is a special case of our notion by first taking $\lambda=1/2$, and then timing the scalar 2, i.e. $\p\oplus\q\prec2\bm{s}(1/2)$.

Let us now consider the DSMUR after flatness process
\begin{align}\label{fds}
\lambda\p\oplus(1-\lambda)\q \prec \mathcal{F} (\bm{s}(\lambda)) \prec \bm{s}(\lambda).
\end{align}
Unlike the case of DPMUR, the vector $\mathcal{F} (\bm{s}(\lambda))$ is optimal and can be calculate explicitly. Moreover, for DSMUR with uniform distribution, i.e. $\lambda=1/2$, one can easily show that $\p\oplus\q \prec 2 \mathcal{F} (\bm{s}(1/2)) \prec 2\bm{s}(1/2)$. Note that, the flatness process cannot be applied to $\p\oplus\q\prec2\bm{s}(1/2)$ directly \cite{Li2019,Wang2019E}, since the results presented in \cite{Cicalese2002} are only designed for probabilities. To accommodate this, a more general lemma is proved in our Supplementary Material \cite{SM}.



\textbf{\emph{Experimental setup.---}}The experimental setup used for verifications of DPMUR and DSMUR is shown in Fig.~\ref{fig:setup}. It consists of single-photon source (see Supplementary Material for details), state preparation, and measurement modules.

In the state preparation module, we prepare a family of $4$-dimensional states with parameters $\theta$ and $\phi$, $|\psi_{\theta,\phi}\rangle=\cos\theta\sin\phi|0\rangle+\cos\theta\cos\phi|1\rangle+\sin\theta|2\rangle+0|3\rangle$,
which is encoded by four modes of a single photon. States $|0\rangle$ and $|1\rangle$ are encoded by different polarizations of the photon in the lower mode, and $|2\rangle$ and $|3\rangle$ are encoded by polarization of the photon in the upper mode. The beam displacer (BD) causes the vertical polarized photons to be transmitted directly, and the horizontal polarized photons to undergo a $4$~mm lateral displacement. When the photon passes through a half-wave plate (H$_{1}$) with a certain setting angle, it is splited by BD1 into two parallel spatial modes -- upper and lower modes. Therefore the photon is prepared in the desired state $|\psi_{\theta,\phi}\rangle$, with parameters $\theta$ and $\phi$ are controlled by the plates H$_{1}$ and H$_{2}$, respectively.


In the measurement module, we consider a setting with a pair of measurements

\begin{equation}\label{2measurement}
\begin{aligned}
&A=\left\{|0\rangle, |1\rangle, |2\rangle, |3\rangle\right\}
\\&B=\begin{aligned}&\{(|0\rangle-i|1\rangle-i|2\rangle+|3\rangle)/2, (|0\rangle-i|1\rangle+i|2\rangle-|3\rangle)/2, \\& (|0\rangle+i|1\rangle-i|2\rangle-|3\rangle)/2,(|0\rangle+i|1\rangle+i|2\rangle+|3\rangle)/2\}\end{aligned}
\end{aligned}
\end{equation}

and another one with multi-measurements
\begin{equation}\label{3measurement}
\begin{aligned}
&C_{1}=\left\{|0\rangle, |1\rangle, |2\rangle, |3\rangle\right\}
\\&C_{2}=\left\{|0\rangle, \frac{|2\rangle+|3\rangle}{\sqrt{2}}, \frac{|1\rangle+|2\rangle-|3\rangle}{\sqrt{3}}, \frac{2|1\rangle-|2\rangle+|3\rangle}{\sqrt{6}}\right\}
\\&C_{3}=\left\{\frac{|2\rangle+|3\rangle}{\sqrt{2}}, |1\rangle, \frac{|0\rangle+|2\rangle-|3\rangle}{\sqrt{3}}, \frac{2|0\rangle-|2\rangle+|3\rangle}{\sqrt{6}}\right\}.
\end{aligned}
\end{equation}

\begin{figure}[t]
\centering
\begin{tikzpicture}[scale=0.75]
  \node[inner sep=0pt] at (0,0) {\includegraphics[width=8.8cm,height=8cm]{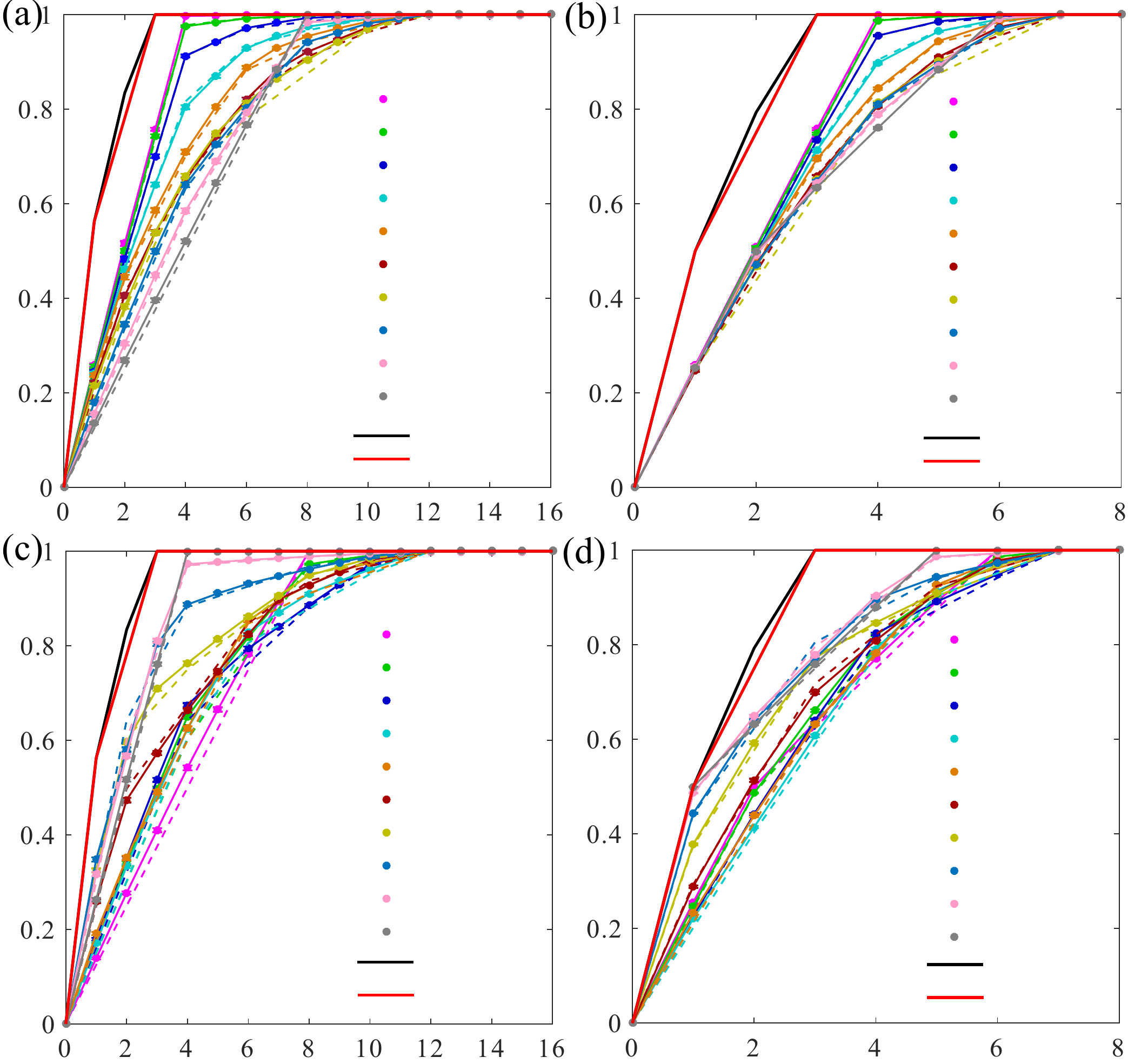}};

  \node[scale=0.8] at (-1,4.35) {\footnotesize $\phi = 0^{\circ}$};
  \node[scale=0.8] at (-0.93,4.02) {\footnotesize $\phi = 10^{\circ}$};
  \node[scale=0.8] at (-0.93,3.69) {\footnotesize $\phi = 20^{\circ}$};
  \node[scale=0.8] at (-0.93,3.36) {\footnotesize $\phi = 30^{\circ}$};
  \node[scale=0.8] at (-0.93,3.03) {\footnotesize $\phi = 40^{\circ}$};
  \node[scale=0.8] at (-0.93,2.70) {\footnotesize $\phi = 50^{\circ}$};
  \node[scale=0.8] at (-0.93,2.37) {\footnotesize $\phi = 60^{\circ}$};
  \node[scale=0.8] at (-0.93,2.04) {\footnotesize $\phi = 70^{\circ}$};
  \node[scale=0.8] at (-0.93,1.71) {\footnotesize $\phi = 80^{\circ}$};
  \node[scale=0.8] at (-0.93,1.38) {\footnotesize $\phi = 90^{\circ}$};
  \node[scale=0.8] at (-0.93,1.05) {\footnotesize $\bm{t}$};
  \node[scale=0.8] at (-0.93,0.72) {\footnotesize $\mathcal{F} ( \bm{t} )$};

  \node[scale=0.8] at (-1,-1) {\footnotesize $\theta = 0^{\circ}$};
  \node[scale=0.8] at (-0.93,-1.33) {\footnotesize $\theta = 10^{\circ}$};
  \node[scale=0.8] at (-0.93,-1.66) {\footnotesize $\theta = 20^{\circ}$};
  \node[scale=0.8] at (-0.93,-1.99) {\footnotesize $\theta = 30^{\circ}$};
  \node[scale=0.8] at (-0.93,-2.32) {\footnotesize $\theta = 40^{\circ}$};
  \node[scale=0.8] at (-0.93,-2.65) {\footnotesize $\theta = 50^{\circ}$};
  \node[scale=0.8] at (-0.93,-2.98) {\footnotesize $\theta = 60^{\circ}$};
  \node[scale=0.8] at (-0.93,-3.31) {\footnotesize $\theta = 70^{\circ}$};
  \node[scale=0.8] at (-0.93,-3.64) {\footnotesize $\theta = 80^{\circ}$};
  \node[scale=0.8] at (-0.93,-3.97) {\footnotesize $\theta = 90^{\circ}$};
  \node[scale=0.8] at (-0.93,-4.30) {\footnotesize $\bm{t}$};
  \node[scale=0.8] at (-0.93,-4.63) {\footnotesize $\mathcal{F} ( \bm{t} )$};

  \node[scale=0.8] at (4.9,4.35) {\footnotesize $\phi = 0^{\circ}$};
  \node[scale=0.8] at (4.97,4.02) {\footnotesize $\phi = 10^{\circ}$};
  \node[scale=0.8] at (4.97,3.69) {\footnotesize $\phi = 20^{\circ}$};
  \node[scale=0.8] at (4.97,3.36) {\footnotesize $\phi = 30^{\circ}$};
  \node[scale=0.8] at (4.97,3.03) {\footnotesize $\phi = 40^{\circ}$};
  \node[scale=0.8] at (4.97,2.70) {\footnotesize $\phi = 50^{\circ}$};
  \node[scale=0.8] at (4.97,2.37) {\footnotesize $\phi = 60^{\circ}$};
  \node[scale=0.8] at (4.97,2.04) {\footnotesize $\phi = 70^{\circ}$};
  \node[scale=0.8] at (4.97,1.71) {\footnotesize $\phi = 80^{\circ}$};
  \node[scale=0.8] at (4.97,1.38) {\footnotesize $\phi = 90^{\circ}$};
  \node[scale=0.8] at (4.97,1.05) {\footnotesize $\bm{s}(1/2)$};
  \node[scale=0.8] at (5.05,0.72) {\footnotesize $\mathcal{F} ( \bm{s}(1/2) )$};

  \node[scale=0.8] at (4.9,-1) {\footnotesize $\theta = 0^{\circ}$};
  \node[scale=0.8] at (4.97,-1.33) {\footnotesize $\theta = 10^{\circ}$};
  \node[scale=0.8] at (4.97,-1.66) {\footnotesize $\theta = 20^{\circ}$};
  \node[scale=0.8] at (4.97,-1.99) {\footnotesize $\theta = 30^{\circ}$};
  \node[scale=0.8] at (4.97,-2.32) {\footnotesize $\theta = 40^{\circ}$};
  \node[scale=0.8] at (4.97,-2.65) {\footnotesize $\theta = 50^{\circ}$};
  \node[scale=0.8] at (4.97,-2.98) {\footnotesize $\theta = 60^{\circ}$};
  \node[scale=0.8] at (4.97,-3.31) {\footnotesize $\theta = 70^{\circ}$};
  \node[scale=0.8] at (4.97,-3.64) {\footnotesize $\theta = 80^{\circ}$};
  \node[scale=0.8] at (4.97,-3.97) {\footnotesize $\theta = 90^{\circ}$};
  \node[scale=0.8] at (4.97,-4.30) {\footnotesize $\bm{s}(1/2)$};
  \node[scale=0.8] at (5.05,-4.63) {\footnotesize $\mathcal{F} ( \bm{s}(1/2) )$};

\end{tikzpicture}
\caption{(color online) Experimental investigation of DPMUR and DSMUR with two measurements. Lorenz curves in (a) and (b) show the experimental datum for DPMUR and DSMUR with states $|\psi_{\pi/4,\phi}\rangle$, and the Lorenz curves in (c) and (d) exhibit the joint uncertainties of DPMURs and DSMURs with states $|\psi_{\theta,\pi/4}\rangle$. Blue curves represent the previous bounds $\bm{t}$ ($\bm{s}(1/2)$), and our improved bounds $\mathcal{F} (\bm{t})$ ($ \mathcal{F} ( \bm{s}(1/2) )$)) are highlighted in red. The dotted lines marked with different colours indicate joint uncertainties with different parameters.}
\label{DSandDP_twoM}
\end{figure}
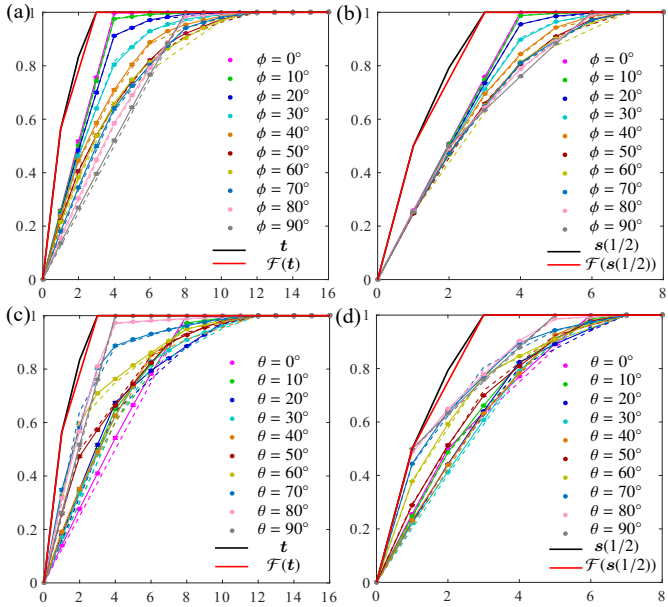

In Fig. (\ref{fig:setup}), device (a) is used to realize measurements $A$ and $C_{1}$. In the presence of quarter-wave plates with an angle of $45^{\circ}$, device (b) is used to realize measurement $B$, and the setting angles of H$_{3}$--H$_{6}$ are $45^{\circ}$, $0^{\circ}$, $22.5^{\circ}$, and $22.5^{\circ}$. On the other hand, in the absence of quarter-wave plates,
device (b) is exploited to implement measurement $C_{2} ( C_{3} )$ when the setting angles of H$_{3}$--H$_{6}$ are $22.5^{\circ}$, $0^{\circ}(45^{\circ})$, $27.4^{\circ}$, and $0^{\circ}$.

\begin{figure}[t]
\begin{tikzpicture}[scale=0.75]
  \node[inner sep=0pt] at (0,0) {\includegraphics[width=8.8cm,height=8.4cm]{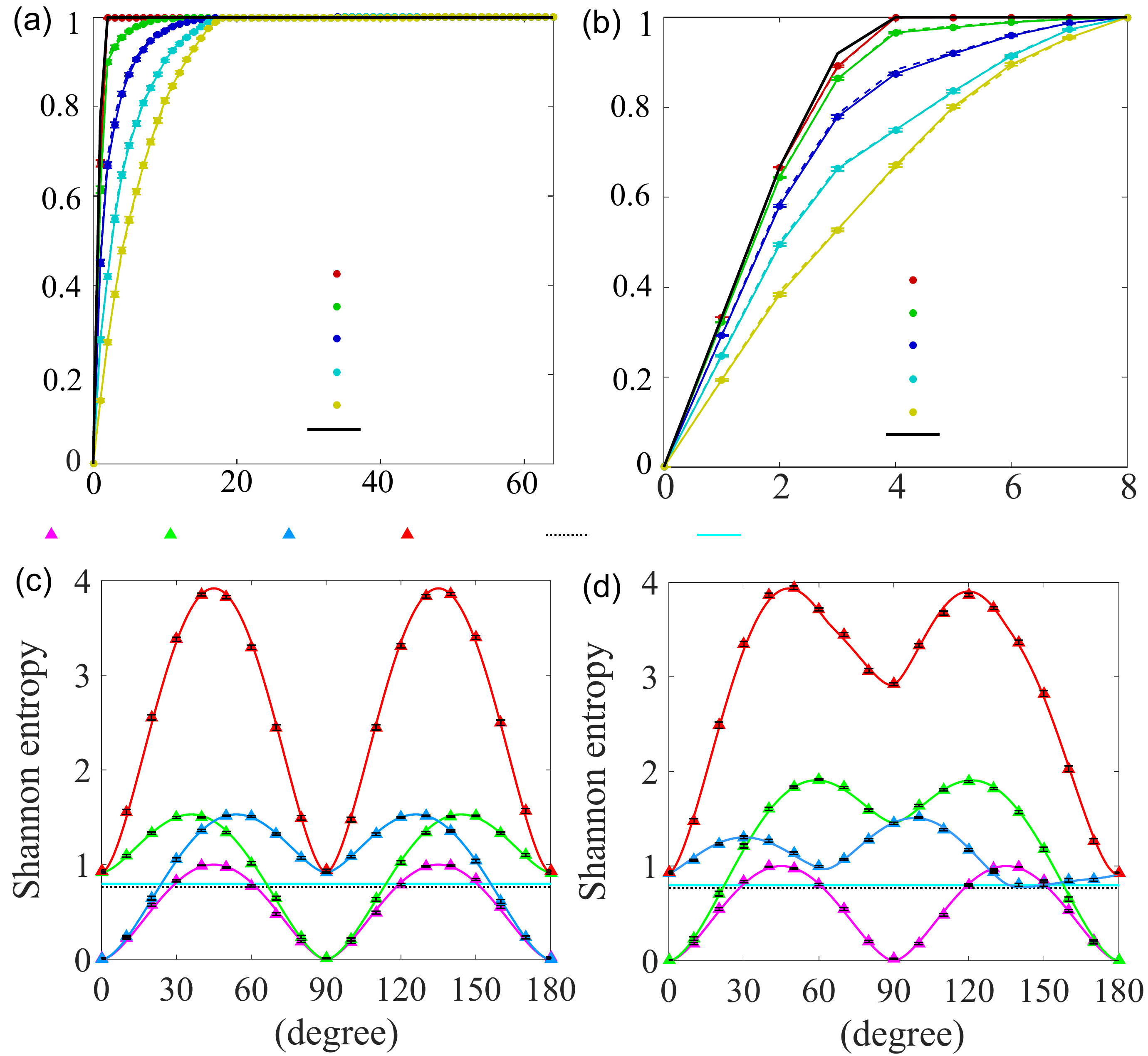}};

  \node[scale=0.8] at (-1.6,2.73) {\footnotesize $\phi = 0^{\circ}$};
  \node[scale=0.8] at (-1.53,2.40) {\footnotesize $\phi = 10^{\circ}$};
  \node[scale=0.8] at (-1.53,2.07) {\footnotesize $\phi = 20^{\circ}$};
  \node[scale=0.8] at (-1.53,1.74) {\footnotesize $\phi = 30^{\circ}$};
  \node[scale=0.8] at (-1.53,1.41) {\footnotesize $\phi = 40^{\circ}$};
  \node[scale=0.8] at (-1.53,1.08) {\footnotesize $\mathcal{F}(\bm{t}^{\prime})$};

  \node[scale=0.8] at (4.26,2.65) {\footnotesize $\phi = 0^{\circ}$};
  \node[scale=0.8] at (4.33,2.32) {\footnotesize $\phi = 10^{\circ}$};
  \node[scale=0.8] at (4.33,1.99) {\footnotesize $\phi = 20^{\circ}$};
  \node[scale=0.8] at (4.33,1.66) {\footnotesize $\phi = 30^{\circ}$};
  \node[scale=0.8] at (4.33,1.33) {\footnotesize $\phi = 40^{\circ}$};
  \node[scale=0.8] at (4.63,1.00) {\footnotesize $\mathcal{F}(\bm{s}^{\prime}(1/3))$};

  \node[scale=0.8] at (-4.8,-0.1) {\footnotesize $H(C_{1})$};
  \node[scale=0.8] at (-3.6,-0.1) {\footnotesize $H(C_{2})$};
  \node[scale=0.8] at (-2.4,-0.1) {\footnotesize $H(C_{3})$};
  \node[scale=0.8] at (-1.0,-0.1) {\footnotesize $\sum_{i} H(C_{i})$};
  \node[scale=0.8] at (0.7,-0.06) {\footnotesize $H ( \mathcal{F}(\bm{t}^{\prime}) )$};
  \node[scale=0.8] at (2.7,-0.06) {\footnotesize $H ( 3\mathcal{F}(\bm{s}^{\prime}(1/3)))$};

  \node[scale=1] at (-3.3,-5.38) {\footnotesize $\phi$};
  \node[scale=1] at (2.45,-5.38) {\footnotesize $\theta$};

\end{tikzpicture}
\caption{(color online) Experimental investigation of DPMURs and DSMURs with three measurements. The plots in (a) and (b) show the joint uncertainties of the quantum state $|\psi_{\pi,\phi}\rangle$ varied with $\phi$ under measurements $C_{1}$, $C_{2}$, $C_{3}$, and our bound $\mathcal{F}(\bm{t}^{\prime})$, $\mathcal{F}(\bm{s}^{\prime}(1/3))$ by means of the Lorenz curves. The plots in (c) and (d) show the Shannon entropic uncertainty relations, with measurements $C_{1}$, $C_{2}$, $C_{3}$, of states $|\psi_{\pi,\phi}\rangle$ and $|\psi_{\theta,\pi/2}\rangle$ respectively. Here the curves marked with magenta, green, and blue stand for the uncertainty associated with measurements $C_{1}$, $C_{2}$ and $C_{3}$; that are $H(C_{1})$, $H(C_{2})$ and $H(C_{3})$, and the curve dyed red represents their joint uncertainties $\sum_{i} H ( C_{i} )$.
The dotted line ($H ( \mathcal{F}(\bm{t}^{\prime}) )=0.7651$) and solid line ($H ( 3\mathcal{F}(\bm{s}^{\prime}(1/3)))=0.7979$) are the bounds of DPMUR and DSMUR. Error bars emphasize the standard deviation of our experimental datum.}
\label{DPDS_shannon}
\end{figure}
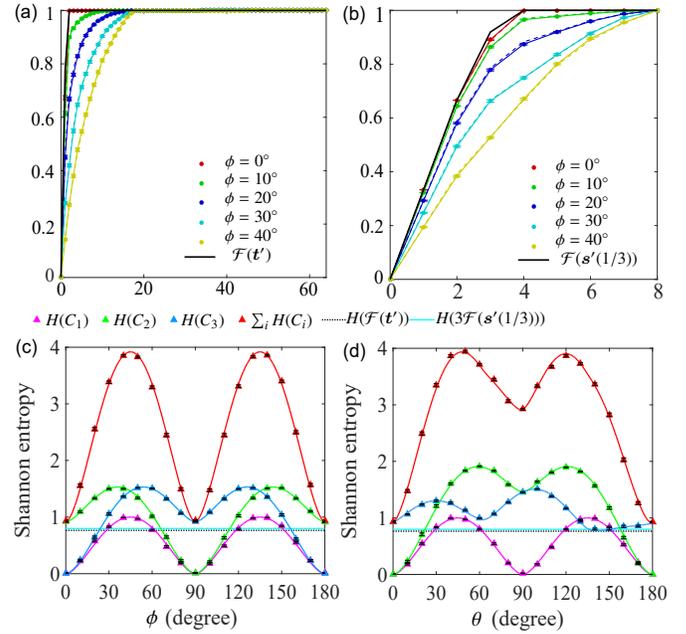

\textbf{\emph{Experimental results.---}}The experimental datum induced by performing measurements (\ref{2measurement}, \ref{3measurement}) on $|\psi_{\theta,\phi}\rangle$ are acquired, and the target of verifying the MURs is fulfilled. In order to unfold the MURs intuitively and geometrically, we employ the technique of {\it Lorentz curve} \cite{Majorization}; for an non-negative vector $\x=(x_{i})_{i=1}^{n}$ with non-increasing order, the corresponding Lorenz curve $\mathcal{L}(\x)$ is defined as the linear interpolation of the points $\{(k,\sum_{i=1}^{k}x_{i})_{k=0}^{n}\}$ with the convention $(0, 0)$ for $k=0$. Based on Lorenz curves, we have $\mathcal{L}(\x)$ lays everywhere below $\mathcal{L}(\y)$ if and only if $\x \prec \y$.

For measurements $A$ and $B$, the bound $\bm{t}$ for DPMUR $\p\otimes\q$, introduced in \cite{PRL,JPA}, is given by $(0.5625, 0.1661, 0.2714)$, and the bound $2\bm{s}(1/2)$ for DSMUR $\p\oplus\q$, introduced in \cite{PRA}, is given by $(0.5, 0.2071, 0.2929)$. To further improve previous results on MURs, we apply the flatness process $\mathcal{F}$ to the bounds $\bm{t}$, $\bm{s}(1/2)$, and acquire $\mathcal{F} (\bm{t})=(0.5625, 0.21875, 0.21875)$, $\mathcal{F} (\bm{s}(1/2)) = (0.5, 0.25, 0.25)$. In Fig. (\ref{DSandDP_twoM}), the dotted lines are obtained by transforming the experimental datum into Lorenz curves. The experimental plots depicted in Fig. (\ref{DSandDP_twoM}) confirm the betterments of our bounds by showing that all experimental datum-induced Lorenz curves lay below our bounds $\mathcal{F} (\bm{t})$ ($\mathcal{F} (\bm{s}(1/2))$), and our bounds are under the previous ones $\bm{t}$ ($\bm{s}(1/2)$).

For measurements $C_{1}$, $C_{2}$ and $C_{3}$, the bound $\mathcal{F}(\bm{t}^{\prime})$ for DPMUR is given by $(0.7773, 0.2227)$, and the bound $\mathcal{F}(\bm{s}^{\prime}(1/3))$ for DSMUR is given by $(1, 1, 0.7583, 0.2417)/3$. In Fig.~\ref{DPDS_shannon} (a) and (b) we see that the joint uncertainties associated with different parameters $\phi$ of the states $|\psi_{\theta,\phi}\rangle$ are marjorized by our bounds $\mathcal{F}(\bm{t}^{\prime})$ and $\mathcal{F}(\bm{s}^{\prime}(1/3))$. Entropies are important tools in both information theory and quantum information theory, and they are closely related to the majorization. From the properties of majorization, it follows the entropic uncertainty relations $\sum_{i} H ( C_{i} ) \geqslant H ( \mathcal{F}(\bm{t}^{\prime}) )$ and $\sum_{i} H ( C_{i} ) \geqslant H ( 3\mathcal{F}(\bm{s}^{\prime}(1/3))) $ with $H$ stands for the {\it Shannon entropy}. All of this can be seen in Fig.~\ref{DPDS_shannon} (c) and (d).


\textbf{\emph{Conclusions.}}--Our guessing game formalism of MURs enable us to classify DPMUR and DSMUR into spatially-separated and temporally-separated joint uncertainties accordingly, which differs from previous developments and, more important, exhibit the essential differences of physical features between DPMUR and DSMUR theoretically. We also implemented an optical experiment that demonstrates the MURs.
In order to present the experimental data efficiently, a novel technique, called Lorenz curve, has been employed. Furthermore, it is advantageous to apply the techniques of flatness process to tighter the bounds of MURs, and its efficiency is confirmed by our experiment. The existence of MURs provides tremendous flexibility in formulating uncertainty relations, and greatly enhance our understanding of quantum mechanics. Therefore, the new formalism, and tighter bounds, as well as the corresponding experimental investigation presented in this work would deeper our knowledge of the quantum world.

{\bf Acknowledgements:}
This work is supported by the National Natural Science Foundation of China (Grants No. 11574291 and No. 11774334), China Postdoctoral Science Foundation (Grant No. 2016M602012 and No. 2018T110618), National Key Research and Development Program of China (Grants No. 2016YFA0301700 and No.2017YFA0304100), and Anhui Initiative in Quantum Information Technologies. Y. Xiao and G. Gour acknowledge NSERC support. S.-M. Fei acknowledges financial support from the National Natural Science Foundation of China under Grant No. 11675113 and Beijing Municipal Commission of Education (KZ201810028042).

\newpage
\onecolumngrid
\begin{center}
\vspace*{.5\baselineskip}
{\textbf{\large Supplemental Material: \\[3pt] Strong Majorization Uncertainty Relations: Theory and Experiment}}
\end{center}

This supplemental material contains a more detailed analysis and extensions of the results presented in the main text. We may reiterate some of the definitions and concepts in the main text to make the supplemental material more explicit and self-contained.

\section{Majorization Lattice}

Before proceeding, it is worth introducing the basic concepts of lattice.

\begin{definition}[Poset]
A partial order is a binary relation ``$\prec$'' over a set $\mathcal{L}$ satisfying reflexivity, antisymmetry, and transitivity. That is, for all $x$, $y$, and $z$ in $\mathcal{L}$, we have
\begin{enumerate}[label=(\roman*)]
\item Reflexivity: $x \prec x$,
\item Antisymmetry: If $x \prec y$ and $y \prec x$, then $x=y$,
\item Transitivity: If $x \prec y$ and $y \prec z$, then $x \prec z$.
\end{enumerate}
\end{definition}
\noindent Note that without the antisymmetry, ``$\prec$'' is just a preorder. Let us now define the set of all $n$-dimensional probability vectors as
\begin{align}
\mathcal{P}^{n} = \left\{\p=\left(p_{1}, \ldots, p_{n}\right)~|~
p_{j}\in[0,1], \sum\limits_{j=1}^{n}p_{j}=1, p_{j}\geqslant p_{j+1}
\right\},
\end{align}
with components in non-increasing order. Accordingly, majorization is a partial order over $\mathcal{P}^{n}$, i.e. $\langle\mathcal{P}^{n}, \prec\rangle$ is a poset.

\begin{definition}[Lattice]
A poset $\langle \mathcal{L}, \prec \rangle$ is called a join-semilattice, if for any two elements $x$ and $y$ of $\mathcal{L}$, it has a unique least upper bound (lub,supremum) $x \lor y$ satisfying
\begin{enumerate}[label=(\roman*)]
\item $x \lor y \in \mathcal{L}$,
\item $x \prec x \lor y$ and $y \prec x \lor y$.
\end{enumerate}
On the other hand, $\langle \mathcal{L}, \prec \rangle$ is called a meet-semilattice, if for any two elements $x$ and $y$ of $\mathcal{L}$, it has a unique greatest lower bound (glb,infimum) $x \land y$ satisfying
\begin{enumerate}[label=(\roman*)]
\item $x \land y \in \mathcal{L}$,
\item $x \land y \prec x$ and $x \land y \prec y$.
\end{enumerate}
$\langle \mathcal{L}, \prec \rangle$ is called a lattice if it is both a join-semilattice and a meet-semilattice, and denote it as a quadruple $\langle \mathcal{L}, \prec, \land, \lor \rangle$.
\end{definition}

\begin{definition}[Complete Lattice]
A lattice $\langle \mathcal{L}, \prec, \land, \lor \rangle$ is called complete, if for any subset $\mathcal{S}\subset\mathcal{L}$, it has a greatest element, denoted by $\top$ and a least element, denoted by $\bot$ which satisfy
\begin{enumerate}[label=(\roman*)]
\item $x \prec \top$, \, for all $x\in\mathcal{S}$ and
$x \prec y$ \, for all $x\in\mathcal{S}$ $\Rightarrow \top \prec y$,
\item $\bot \prec x$, \, for all $x\in\mathcal{S}$ and
$y \prec x$ \, for all $x\in\mathcal{S}$ $\Rightarrow y \prec \bot$.
\end{enumerate}
\end{definition}

\noindent By embedding the majorization ``$\prec$'', the quadruple $\langle \mathcal{P}^{n}, \prec, \land, \lor \rangle$ forms a complete lattice. We remark that the result of completeness follows directly from the work presented in \cite{Rapat1991SM}, and the algorithm in finding the greatest element and the least element of a subset $\mathcal{S}$ (also known as flatness process) was first introduced in \cite{Cicalese2002SM}. As we are trying to connect the structure of majorization lattice with MURs, here we are only interested in the construction of the greatest element of $\mathcal{S} \subset \mathcal{P}^{n}$. Roughly speaking, there are two steps in finding it; that are

\begin{itemize}
\item Step 1: Finding the largest partial sums; for each $\x = (x_{1}, \ldots, x_{n}) \in \mathcal{S}$, we need to evaluate the following quantities
\begin{align}
Y_{k} := \max_{\x \in \mathcal{S}}\sum_{i=1}^{k} x_{i},
\end{align}
and collect these numbers into a vector $\y := (Y_{1}, Y_{2} - Y_{1}, \ldots, Y_{n} - Y_{n-1}) := (y_{1}, y_{2}, \ldots, y_{n})$. Clearly, we have $\x \prec \y$ for all $\x \in \mathcal{S}$. From now on, we denote the vector $\y$ as $\vee \mathcal{S}$; that is $\y := \vee \mathcal{S}$, and $\x \prec \vee \mathcal{S}$ for all $\x \in \mathcal{S}$. For the set with finite elements, such as $\mathcal{S} = \{ \x_{1}, \ldots, \x_{k}\}$, we can also use $\x_{1}\vee \ldots \vee \x_{k}$ to stand for $\vee \mathcal{S}$.

\item Step 2: Flatness process; let $j$ be the smallest integer in $\left\{2, \ldots, n\right\}$ such that $y_{j}>y_{j-1}$, and $i$ be the greatest integer in $\left\{1, \ldots, j-1\right\}$ such that $y_{i-1} \geqslant (\sum_{k=i}^{j} y_{k})/(j-i+1):=a$. Define
%
\begin{align}\label{eq; mj bound t}
\mathcal{F} ( \y ) := \left(F_{1}, \ldots, F_{n}\right) \quad \text{with} \quad F_{k} =
     \begin{cases}
       a & \text{for}\quad k = i, \ldots, j \\
       y_{k} & \text{otherwise.} \\
     \end{cases}
\end{align}
Here we also use the notation $\mathcal{F} ( \vee \mathcal{S} )$ to denote $\mathcal{F} ( \y )$.
\end{itemize}
Rigorously speaking, to obtain the optimal bounds from a set $\mathcal{S}$, which contains infinite number of elements, by only applying Steps 1 and 2 is far from enough \cite{Li2019,Wang2019E}. Actually, we should first guarantee the target set $\mathcal{S}$ is a subset of some complete lattice, which ensures the existence of the optimal bounds.

A key lemma in proving the optimality is the following lemma, which was first proved in \cite{Cicalese2002SM}
\begin{lem}\label{lem1}
Let $\x$, $\y \in \mathcal{P}^{n}$, there exists a unique optimal upper bound $\mathcal{F} (\x \vee \y) \in \mathcal{P}^{n}$, satisfying
\begin{itemize}
\item $\x \prec \mathcal{F} (\x \vee \y)$, and $\y \prec \mathcal{F} (\x \vee \y)$;
\item for any $\z \in \mathcal{P}^{n}$ satisfying $\x \prec \z$ and $\y \prec \z$, it follows $\mathcal{F} (\x \vee \y) \prec \z$.
\end{itemize}
\end{lem}
The domain of Lemma \ref{lem1} is the set $\mathcal{P}^{n}$. However, its generalization is also correct. Let us now consider the following set
\begin{align}
\mathcal{P}^{n}_{c} = \left\{\p=\left(p_{1}, \ldots, p_{n}\right)~|~
p_{j} \geqslant 0, \sum\limits_{j=1}^{n}p_{j}=c, p_{j}\geqslant p_{j+1}
\right\},
\end{align}
with a constant $c$. Then we have
\begin{lem}\label{lem2}
Let $\x$, $\y \in \mathcal{P}^{n}_{c}$, there exists a unique optimal upper bound $\mathcal{F} (\x \vee \y) \in \mathcal{P}^{n}_{c}$, satisfying
\begin{itemize}
\item $\x \prec \mathcal{F} (\x \vee \y)$, and $\y \prec \mathcal{F} (\x \vee \y)$;
\item for any $\z \in \mathcal{P}^{n}_{c}$ satisfying $\x \prec \z$ and $\y \prec \z$, it follows $\mathcal{F} (\x \vee \y) \prec \z$.
\end{itemize}
\end{lem}

\begin{proof}
By dividing the constant $c$, we obtain
\begin{align}
\frac{1}{c}\x \in \mathcal{P}^{n}, \quad \frac{1}{c}\y \in \mathcal{P}^{n},
\end{align}
which implies
\begin{align}
\frac{1}{c}\x \prec \mathcal{F} ( \frac{1}{c}\x \vee \frac{1}{c}\y ), \quad \frac{1}{c}\y \prec \mathcal{F} ( \frac{1}{c}\x \vee \frac{1}{c}\y ).
\end{align}
Then for the positive vector $\x$ and $\y$ we have that
\begin{align}
\x \prec c \, \mathcal{F} ( \frac{1}{c}\x \vee \frac{1}{c}\y ), \quad
\y \prec c \, \mathcal{F} ( \frac{1}{c}\x \vee \frac{1}{c}\y ).
\end{align}
Now due to the fact that $\mathcal{F}$ is scalar-multiplication-preserving, we get $c \, \mathcal{F} ( \frac{1}{c}\x \vee \frac{1}{c}\y ) = \mathcal{F} (\x \vee \y)$, and hence
\begin{align}
\x \prec \mathcal{F} (\x \vee \y), \quad \y \prec \mathcal{F} (\x \vee \y).
\end{align}
For any $\z \in \mathcal{P}^{n}_{c}$ satisfying $\x \prec \z$ and $\y \prec \z$, we have
\begin{align}
\mathcal{F} ( \frac{1}{c}\x \vee \frac{1}{c}\y ) \prec \frac{1}{c} \z,
\end{align}
which immediately yields $\mathcal{F} (\x \vee \y) \prec \z$ and completes the proof.
\end{proof}
As a corollary of our Lemma \ref{lem2}, previous statement remains valid when the domain has been replaced by the set $S = \{\p=\left(p_{1}, \ldots, p_{n}\right)~|~
p_{j}\in[0,1], \sum_{j=1}^{n}p_{j}=c, p_{j}\geqslant p_{j+1}
\}$, and this proves the key lemma used in \cite{Li2019}. Here we only show the proof of two elements, but actually it works for any countable elements \cite{Bosyk2019SM}.

\section{Bounds for DPMUR}

Now we are in the position to construct the optimal for DPMUR. Note that the set of spatially-separated joint uncertainty $\p\otimes\q$ forms a subset of $\mathcal{P}^{n}$, i.e. here $\mathcal{S} = \{ \p\otimes\q \} \subset \mathcal{P}^{n}$. From Step 1, we have $Y_{k} = R_{k}$ with $R_{k}$ defined in the main text. For the collection of quantities $R_{k}$, we apply the flatness process and obtain $\mathcal{F} (\bm{r})$. Therefore, we have $\p\otimes\q \prec \mathcal{F} (\bm{r})$ for all probability vector $\p$ and $\q$, and $\mathcal{F} (\bm{r})$ is the largest element for $\mathcal{S} = \{ \p\otimes\q \}$, and hence optimal.

However, the vector $\bm{r}$ can be in general difficult to calculate explicitly, as they involve a complicated optimization problem. Fortunately, we still have the following relaxing method,
\begin{align}
R_{k} := \max\limits_{I_{k}}\max\limits_{\rho}\sum\limits_{(a,b)\in I_{k}}
\p_{a}(\rho)~\q_{b}(\rho)
\leqslant \max\limits_{I_{k_{1}}, I_{k_{2}}}
\max\limits_{\rho} \left( \sum\limits_{a\in I_{k_{1}}} \p_{a}(\rho) \right) \left( \sum\limits_{b\in I_{k_{2}}} \p_{b}(\rho) \right)
\leqslant  \max\limits_{I_{k_{1}}, I_{k_{2}}} \max\limits_{\rho}
\left( \frac{\sum_{a\in I_{k_{1}}} \p_{a}(\rho) + \sum_{b\in I_{k_{2}}} \p_{b}(\rho) }{2}\right)^{2},
\end{align}
with
\begin{align}
\max\limits_{\rho}
\left( \frac{\sum_{a\in I_{k_{1}}} \p_{a}(\rho) + \sum_{b\in I_{k_{2}}} \p_{b}(\rho) }{2}\right)^{2}
= \left( \frac{ \lambda_{1}( \sum_{a\in I_{k_{1}}} M_{a} + \sum_{b\in I_{k_{2}}} N_{b} ) }{2} \right)^{2},
\end{align}
and their indices $k_{1}$ and $k_{2}$ satisfying $k_{1} + k_{2} = k+1$. Let us define $T_{k}$ as
\begin{align}
T_{k} &:= \max\limits_{I_{k_{1}}, I_{k_{2}}} \left( \frac{ \lambda_{1}( \sum_{a\in I_{k_{1}}} M_{a} + \sum_{b\in I_{k_{2}}} N_{b} ) }{2} \right)^{2},\notag\\
t_{k} &:= T_{k} - T_{k-1}, \notag\\
\bm{t} &:= (t_{1}, \ldots, t_{mn}).
\end{align}
Note that here the vector $\bm{t}$ can be computed explicitly, satisfying the following inequalities
\begin{align}\label{UUR}
\p\otimes\q \prec \bm{r} \prec \bm{t}.
\end{align}
Eq. (\ref{UUR}) is the main result of \cite{PRL}, which is also the implementation of Step 1 presented in the previous section. In order to obtain a better bound, the flatness process $\mathcal{F}$ is needed; that is
\begin{align}\label{improvedUUR}
\p\otimes\q \prec \mathcal{F} ( \bm{r} ) \prec \bm{r} \prec \mathcal{F} ( \bm{t} ) \prec \bm{t}.
\end{align}
The proof of (\ref{improvedUUR}) follows \cite{Cicalese2002SM} straightforwardly.

\section{Bounds for DSMUR}

In the cases of (weighted) DSMUR, we have $\mathcal{S} = \{ \lambda\p\oplus (1-\lambda)\q \} \subset \mathcal{P}^{n}$. From Step 1, we have
\begin{align}
S_{k} &:= \max\limits_{|I|+|J|=k}\max\limits_{\rho}\sum\limits_{\substack{a\in I\\b\in J}} \left(\lambda\p_{a}(\rho)+(1-\lambda)\q_{b}(\rho)\right)
= \max\limits_{|I|+|J|=k} \max\limits_{\rho} \Tr \left[\rho \left( \sum\limits_{\substack{a\in I\\b\in J}} (\lambda M_{a} + (1-\lambda) N_{b} ) \right) \right]
= \max\limits_{|I|+|J|=k} \lambda_{1} \left( \sum\limits_{\substack{a\in I\\b\in J}} (\lambda M_{a} + (1-\lambda) N_{b} ) \right) \notag\\
&= \max\limits_{|I|=k}\sum\limits_{c\in I}
\lambda_{1}(G_{c}(\lambda)).
\end{align}
Unlike the cases of DPMUR, here the quantities $S_{k}$ can be computed explicitly. Based on these notations, we construct $\bm{s}(\lambda)$ as $(S_{1}(\lambda), S_{2}(\lambda)-S_{1}(\lambda), \ldots, S_{m+n}(\lambda)-S_{m+n-1}(\lambda))$, which meets the following relation
\begin{align}
\lambda \p \oplus (1-\lambda) \q \prec \bm{s}(\lambda).
\end{align}
Applying the flatness process, we immediately obtain
\begin{align}\label{FDSMUR}
\lambda \p \oplus (1-\lambda) \q \prec \mathcal{F} ( \bm{s}(\lambda) ) \prec \bm{s}(\lambda).
\end{align}
Again, the optimality of $\mathcal{F} ( \bm{s}(\lambda) )$ follows from the completeness of $\mathcal{P}^{n}$ and the flatness process $\mathcal{F}$ directly, not just because of the flatness process \cite{Li2019,Wang2019E}. For random number generator $R$ with uniform distribution, i.e. $\lambda = 1/2$, (\ref{FDSMUR}) implies that
\begin{align}
\frac{1}{2} \p \oplus \frac{1}{2} \q \prec \mathcal{F} ( \bm{s}(1/2) ) \prec \bm{s}(1/2),
\end{align}
and hence we have
\begin{align}
\p \oplus \q \prec 2 \mathcal{F} ( \bm{s}(1/2) ) \prec 2 \bm{s}(1/2).
\end{align}
Note that, the flatness process cannot be applied to the DSMUR $\p\oplus\q\prec2\bm{s}(1/2)$ directly \cite{Li2019,Wang2019E}, since the results presented in \cite{Cicalese2002SM} are only designed for the vector belongs to $\mathcal{P}^{n}$. Otherwise, an appropriate modification of the proof, i.e. our Lemma \ref{lem2}, is needed. This is another reason, from mathematical viewpoints, why our forms of DSMUR are valuable.

\section{Bounds for Multi-measurements MURs}

Uncertainty relation is not the patent of two measurements, so what to make of this? We checked in with a multi-measurements MURs to meake more sense of the ruling. First, we consider DPMUR with multi-measurements. Assume we have a set of POVMs $\{ M_{x} \}_{x=1}^{n}$ with $M_{x} = \{ M_{a|x} \}_{a=1}^{d}$, and the denote outcome probability distribution as $p(a(x)|x) := \Tr [\rho M_{a|x}]$. By collecting these numbers into the probability vectors, we have $p_{x} := (p(a(x)|x))_{a}$, and their spatially-separated joint uncertainty becomes $\bigotimes_{x} p_{x}$. In order to obtain a computing-friendly bound, we apply the Geometric-Arithmetic mean inequality, i.e.
\begin{align}
\max\limits_{I_{k}}\max\limits_{\rho}\sum\limits_{(a(x))_{x}\in I_{k}}
\prod\limits_{x} p(a(x)|x)
\leqslant \max\limits_{\sum_{x} I_{x} = k}\max\limits_{\rho}
\prod\limits_{x}
\sum\limits_{a(x) \in I_{x}} p(a(x)|x)
\leqslant \max\limits_{\sum_{x} I_{x} = k}\max\limits_{\rho}
\left( \frac{\sum_{x} \sum_{a(x) \in I_{x}} p(a(x)|x)}{n}\right)^{n}
= \max\limits_{\sum_{x} I_{x} = k}
\left( \frac{ \lambda_{1}(\sum_{x} \sum_{a(x) \in I_{x}} M_{a|x})}{n}\right)^{n}.
\end{align}
Similarly, define
\begin{align}
T_{k}^{\prime} &:= \max\limits_{\sum_{x} I_{x} = k}
\left( \frac{ \lambda_{1}(\sum_{x} \sum_{a(x) \in I_{x}} M_{a|x})}{n}\right)^{n},\notag\\
t_{k}^{\prime} &:= T_{k}^{\prime} - T_{k-1}^{\prime}, \notag\\
\bm{t}^{\prime} &:= (t_{1}^{\prime}, \ldots, t_{d^{n}}^{\prime}),
\end{align}
which satisfying the following multi-measurements DPMUR
\begin{align}
\bigotimes_{x} p_{x} \prec \bm{t}^{\prime}.
\end{align}
Moreover, by apply the flatness process $\mathcal{F}$ again, we obtain a tighter bound $\mathcal{F} (\bm{t}^{\prime})$; that is
\begin{align}
\bigotimes_{x} p_{x} \prec \mathcal{F} (\bm{t}^{\prime}) \prec \bm{t}^{\prime}.
\end{align}
Remark that the bound $\mathcal{F} (\bm{t}^{\prime})$ outperforms the one constructed in \cite{PRL}, i.e $\bm{t}$. On the other hand, for probability vectors $p_{x}$, we can also consider their joint uncertainties in temporally-separated forms, i.e. $\bigoplus_{x} c_{x} p_{x}$ with $\bm{c} := (c_{x})_{x}$ a probability vector. To find its bound, consider the following equation
\begin{align}
\max\limits_{\sum_{x} I_{x} = k}\max\limits_{\rho}
\sum\limits_{x}
\sum\limits_{a(x) \in I_{x}} c_{x} p(a(x)|x)
= \max\limits_{\sum_{x} I_{x} = k}
\lambda_{1} ( \sum\limits_{x}
\sum\limits_{a(x) \in I_{x}} c_{x} M_{a(x)|x} ) := S_{k}^{\prime},
\end{align}
and define $\bm{s}^{\prime}$ as $(S_{1}^{\prime}, S_{2}^{\prime}-S_{1}^{\prime}, \ldots, S_{nd}^{\prime}-S_{nd-1}^{\prime})$. Based on these notations, we have the following multi-measurements DSMUR
\begin{align}
\bigoplus_{x} c_{x} p_{x} \prec \bm{s}^{\prime}.
\end{align}
The optimal bounds of $\bigoplus_{x} c_{x} p_{x}$ is obtained as
\begin{align}
\bigoplus_{x} c_{x} p_{x} \prec \mathcal{F}(\bm{s}^{\prime}) \prec \bm{s}^{\prime},
\end{align}
by performing the flatness process $\mathcal{F}$. Therefore the construction of the optimal bound for $\bigoplus_{x} p_{x}$ is also straightforward.

\section{Mathematical Comparisons between DPMUR and DSMUR}

With the majorization relation for vectors, we now present DPMUR and DSMUR as
\begin{align}
\p\otimes\q & \prec \bm{t} := \x, \\
\p\oplus\q & \prec 2\bm{s}(1/2) := \y,
\end{align}
where $\rho$ runs over all quantum states in Hilbert space $\mathcal{H}$ with $\x$, $\y$ standing for the state-independent bound of DPMUR and DSMUR respectively. Let us take any nonnegative Schur-concave function $\mathcal{U}$ to quantify the uncertainties and apply it to DPMUR and DSMUR, which leads to
\begin{align}
\mathcal{U}(\p\otimes\q) &\geqslant \mathcal{U}( \x ),\\
\mathcal{U}(\p\oplus\q) &\geqslant \mathcal{U}( \y ).
\end{align}
The universality of MURs comes from the diversity of uncertainty measures $\mathcal{U}$ and DPMUR, DSMURs stand for different kind of uncertainties.

We next move to describe the additivity of uncertainty measures, and call a measure $\mathcal{U}$ {\it direct-product additive} if $\mathcal{U}(\p\otimes\q)=\mathcal{U}(\p) + \mathcal{U}(\q)$. Instead of direct-product between probability distribution vectors, one can also consider direct-sum and define {\it direct-sum additive} for $\mathcal{U}$ whenever it satisfies $\mathcal{U}(\p\oplus\q)=\mathcal{U}(\p) + \mathcal{U}(\q)$. Note that the joint uncertainty $\p\oplus\q$ considered here is unnormalized and comparison between DPMUR and normalized DSMUR is detailed later. Once an uncertainty measure $\mathcal{U}$ is evolved to both direct-product additive and direct-sum additive, then we call it {\it super additive} for uncertainties. It is worth to mention that $\mathcal{U}(\p\otimes\q)=\mathcal{U}(\p\oplus\q)$ whenever the uncertainty measure is super additive.
Consequently, the bound $\y$ for DSMUR performs better than $\x$ in the case of super additive,
\begin{align}\label{eqs}
\mathcal{U}(\p\otimes\q)=\mathcal{U}(\p\oplus\q) \geqslant \mathcal{U}( \y ) \geqslant \mathcal{U}( \x ),
\end{align}
since $\y \prec \left\{1\right\} \oplus \x$ \cite{PRA}. We remark that the well known Shannon entropy is super additive and only by applying super additive functions, like Shannon entropy, DPMUR and DSMUR are comparable. It should also be clear that DPMUR and DSMUR have been employed to describe different type of uncertainties. For an uncertainty measure $\mathcal{U}$, in general, it can be checked that $\mathcal{U}(\p\otimes\q) \neq \mathcal{U}(\p\oplus\q)$ and hence it is meaningless to state that DSMUR performs better than DPMUR and vice versa.

One of the main goals in the study of uncertainty relations is the quantification of the joint uncertainty of incompatible observables. DPMUR and DSMUR provide us two different methods to quantify joint uncertainty between incompatible observables. Relations between DPMUR and DSMUR are of fundamental importance both for the theoretical characterization of joint uncertainties, as well as the experimental implementation. Quite uncannily, we find that for some eligible uncertainty measure $\mathcal{U}$, DPMUR and DSMUR are given by
\begin{align}\label{eq1}
\mathcal{U}(\p\oplus\q) \geqslant \mathcal{U}( \y ) > \mathcal{U}(\p\otimes\q) \geqslant \mathcal{U}( \x ),
\end{align}
for some quantum state $\rho$.

Let us now construct such uncertainty measure $\mathcal{U}$. First define the summation function $\mathcal{S}$ as $\mathcal{S}(\mathbf{u}):=\sum_{l}u_{l}=\left\lVert \mathbf{u} \right\rVert_{1}$ with $\mathbf{u}=(u_{1}, u_{2}, \ldots, u_{d})$. Another important function $\mathcal{M}$ is defined as $\mathcal{M}(\mathbf{u}):=\max_{l}u_{l}=2^{-H_{\text{min}}(\mathbf{u})}$. And hence it is easy to check that $\mathcal{U}:=\mathcal{S}-\mathcal{M}$ is a nonnegative Schur-concave functions; take two vectors satisfying $x \prec y$, and based on the definition of $\mathcal{U}$ we have $\mathcal{U}(x)=\sum^{d}_{j=2}x_{j}^{\downarrow}\geqslant \sum^{d}_{j=2}y_{j}^{\downarrow}=\mathcal{U}(y)$. Specifically this function, which combines $\mathcal{S}$ and $\mathcal{M}$ together, is a qualified uncertainty measure and satisfies Eq. (\ref{eq1}) for some quantum states and measurements. Moreover, specific examples are given in the following experimental demonstration.

In principle, DPMUR and DSMUR do not have to be comparable and their joint uncertainty can be quantified by their bound. However, we can compare their differences by checking which bound approximates their joint uncertainty better since joint uncertainties are often classified by their bounds.
Take any nonnegative Schur-concave function $\mathcal{U}$, which leads to two nonnegative quantities $\xi_{DS}:=\mathcal{U}(\p\oplus\q) - \mathcal{U}( \y )$ and $\xi_{DP}:=\mathcal{U}(\p\otimes\q) - \mathcal{U}( \x )$. To determine whether the bound $\x$ approximates DPMUR better than $\y$ approximates DSMUR, we simply compare the numerical value of $\xi_{DS}$ and $\xi_{DP}$. And how such bounds contribute to the joint uncertainties are depicted in our experiment.

The above discussion on DSMUR is based on its unnormalized form $\p \oplus \q$, since it was first given in \cite{PRA} with the form $\p\oplus\q \prec \y$ for probability distributions $\p$ and $\q$. However, unlike $\p\otimes\q$ constructed in DPMUR \cite{PRL}, $\p\oplus\q$ is not even a probability distribution. In order to derive a normalized DSMUR, we simply take the weight $1/2$
\begin{align}
\frac{1}{2}\p\oplus\frac{1}{2}\q & \prec \frac{1}{2}\y.
\end{align}
And now we compare the normalized DSMUR $\frac{1}{2}\p\oplus\frac{1}{2}\q \prec \frac{1}{2}\y$ with DPMUR $\p\otimes\q \prec \x$; by taking the quantum states shown in the main text
\begin{equation}
\begin{aligned}
|\psi_{\theta,\phi}\rangle&=\cos\theta\sin\phi|0\rangle+\cos\theta\cos\phi|1\rangle+\sin\theta|2\rangle
\\&=(\cos\theta\sin\phi,\cos\theta\cos\phi,\sin\theta,0)^{\top},
\end{aligned}
\end{equation}
and measurements $A$, $B$ with the following eigenvectors
\begin{equation}
\begin{aligned}
&A=\left\{|0\rangle, |1\rangle, |2\rangle, |3\rangle\right\}\\
&B=\left\{\frac{|0\rangle-i|1\rangle-i|2\rangle+|3\rangle}{2}, \frac{|0\rangle-i|1\rangle+i|2\rangle-|3\rangle}{2}, \frac{|0\rangle+i|1\rangle-i|2\rangle-|3\rangle}{2},\frac{|0\rangle+i|1\rangle+i|2\rangle+|3\rangle}{2}\right\}.
\end{aligned}
\end{equation}
We depicted the pictures of $H\left(\p\otimes\q\right)$, $H\left(\frac{1}{2}\p\oplus\frac{1}{2}\q\right)$, $H\left(\x\right)$, and $H\left(\frac{1}{2}\y\right)$ in Fig.~\ref{shannontwo}.

\begin{figure*}[tbph]
\includegraphics [width=16cm,height=6.4cm]{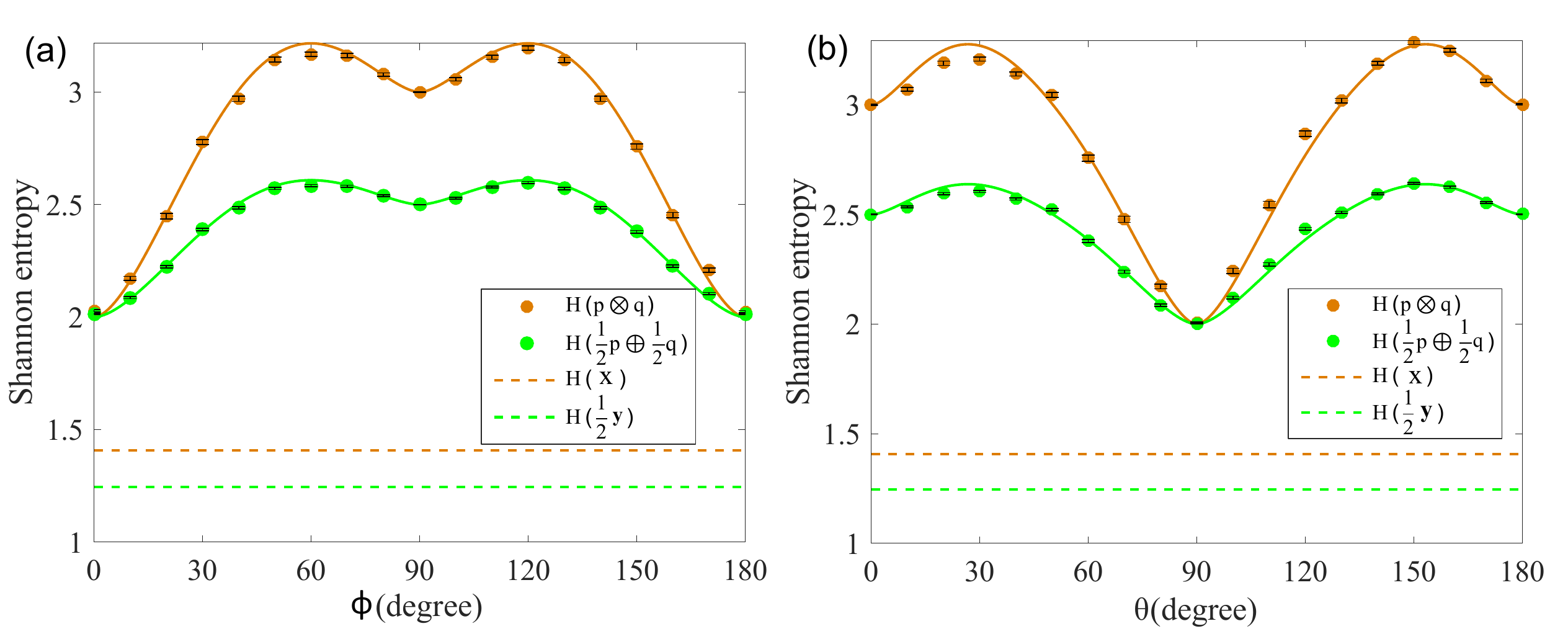}
\caption{Experimental results for the comparison between normalized DSMUR and DPMUR based on the Shannon entropy. Plots in (a) and (b) show the results for measuring states $|\psi_{\pi/4,\phi}\rangle$ and $|\psi_{\theta,\pi/4}\rangle$ with $A$ and $B$, respectively.}
\label{shannontwo}
\end{figure*}

\section{Super Additivity}

To be comparable for DPMUR and DSMUR, we should choose an uncertainty measure $\mathcal{U}$ that are both Schur-concave and super additive. Clearly Shannon entropy is a qualified candidate. The question, thus, naturally arises: is there another function satisfies the following properties:
\begin{description}[align=left]
\item [Property 1] $\mathcal{U}$ should be continuous in $\p$ and $\q$.
\item [Property 2] $\mathcal{U}$ should be a Schur-concave function.
\item [Property 3] $\mathcal{U}$ should be super additive, i.e.
	\begin{align}
        \mathcal{U}(\p\otimes\q)&=\mathcal{U}(\p) + \mathcal{U}(\q),\\
        \mathcal{U}(\p\oplus\q)&=\mathcal{U}(\p) + \mathcal{U}(\q).
        \end{align}
\end{description}
Or will these properties lead to a unique function (up to a scalar)? Since we can take $\q$ as $\left(1, 0, \ldots, 0\right)$, and then $\mathcal{U}(\p\otimes\q)= \mathcal{U}(\p)$ which is continuous in the $p_{i}$ while $\p=\left(p_{i}\right)_{i}$. Moreover, due to the Schur-concavity, $\mathcal{U}$ is a monotonic increasing function of $d$ when taking $p_{i}=\frac{1}{d}$. In addition, if $\mathcal{U}$ complies with the composition law for compound experiments, then there is only one possible expression for $\mathcal{U}$, i.e. Shannon entropy (up to a scalar). Namely, if there is a measure, say $\mathcal{U}(\p)=\mathcal{U}\left(p_{1}, p_{2}, \ldots, p_{d}\right)$ which is required to meet the following three properties:
\begin{description}[align=left]
\item [Property 4] $\mathcal{U}$ should be continuous in $\p$.
\item [Property 5] If all the $p_{i}$ are equal, $p_{i}=\frac{1}{d}$, then $\mathcal{U}$ should be a monotonic increasing function of $d$. With equally $d$ likely events there is more choice, or uncertainty, when there are more possible events.
\item [Property 6 (Composition Law)] If a choice be broken down into two successive choices, the original $\mathcal{U}$ should be the weighted sum of the individual values of $\mathcal{U}$.
\end{description}
Then the only $\mathcal{U}$ satisfying the three above assumptions is of the form \cite{Shannon1948}:
\begin{align}
\mathcal{U}(\p)=k\cdot\left(-\sum\limits_{i=1}^{d}p_{i}\log p_{i}\right),
\end{align}
where $k$ is a positive constant. Whenever a function $\mathcal{U}$ satisfies Property 1 and Property 2, it will meet Property 4 and Property 5 automatically. However, super additivity differs with the Composition Law, and this leads to function satisfied Property 1, 2, and 3 other than Shannon entropy.

For example, consider the composition between logarithmic function and elementary symmetric function:
\begin{align}
\mathcal{V}(\p):=\log\left(\prod\limits_{i=1}^{d}p_{i}\right).
\end{align}
Here $\mathcal{V}$ satisfies Properties 1, 2, and the DPMUR is read as
\begin{align}
\mathcal{V}(\p\otimes\q)&=\log\left(\prod\limits_{i, j}p_{i}q_{j}\right)\notag\\
&=\log\left(\prod\limits_{i}p_{i}\cdot\prod\limits_{j}q_{j}\right)\notag\\
&=\log\left(\prod\limits_{i}p_{i}\right)+\log\left(\prod\limits_{j}q_{j}\right)\notag\\
&=\mathcal{V}(\p)+\mathcal{V}(\q),
\end{align}
where the probability distributions $\p$ and $\q$ are defined as $\left(p_{i}\right)_{i}$ and $\left(q_{j}\right)_{j}$. On the other hand, DSMUR is written as
\begin{align}
\mathcal{V}(\p\oplus\q)&=\log\left(\prod\limits_{i, j}p_{i}q_{j}\right)=\mathcal{V}(\p\otimes\q),
\end{align}
hence, $\mathcal{V}$ meets Property 3. To summarize, we derive a function $\mathcal{V}$, which is valid for Properties 1, 2, and 3. However $\mathcal{V}$ is not a good uncertainty measure, since $\mathcal{V}(\y)$ and $\mathcal{V}(\x)$ are not well defined (due to the occurrence of $\log0$). Whether there exists another function that obeys Properties 1, 2, and 3 remains an open question, and one may conjecture that $\mathcal{V}$, Shannon entropy $H$ and the convex combinations of $\mathcal{V}$ and $H$ are the only suitable candidates.

\end{document}